\documentclass[a4paper,12pt]{amsart}
\usepackage[margin=2cm]{geometry}
\usepackage{pstricks,pst-node,pst-tree,pst-plot}

\newtheorem{theorem}{Theorem}

\newtheorem{lemma}[theorem]{Lemma}

\theoremstyle{definition}
\newtheorem{definition}[theorem]{Definition}
\newtheorem{remark}[theorem]{Remark}
\newtheorem{example}[theorem]{Example}

\title{Multi-dimensional sets recognizable in all abstract numeration systems}
\author{\'Emilie Charlier, Anne Lacroix, Narad Rampersad}
\address{Department of Mathematics \\
  University of Li{\`e}ge \\
  Grande traverse 12 (B37) \\
  B-4000 Li{\`e}ge \\
  Belgium
}

\email[\'E.~Charlier]{echarlier@ulg.ac.be}
\email[A.~Lacroix]{A.Lacroix@ulg.ac.be}
\email[N.~Rampersad]{nrampersad@ulg.ac.be}

\DeclareMathOperator{\rep}{rep}
\DeclareMathOperator{\val}{val}
\DeclareMathOperator{\N}{\mathbb N}

\begin{document}
\begin{abstract}
We prove that the subsets of $\N^d$ that are $S$-recognizable
for all abstract numeration systems $S$ are exactly the
$1$-recognizable sets.  This generalizes a result of Lecomte and Rigo
in the one-dimensional setting.
\end{abstract}
\maketitle

\section{Introduction}
In this paper we characterize the subsets of $\N^d$ that are
simultaneously recognizable in all abstract numeration systems
(numeration systems that represent a natural number $n$ by the
$(n+1)$-th word of a genealogically ordered regular language---see
below for the precise definition).  Lecomte and Rigo \cite{LR01}
provided such a characterization for the case $d=1$ based on the
well-known correspondence between unary regular languages and
ultimately periodic subsets of $\N$.  When $d>1$ we no longer have
such a nice correspondence and the situation becomes somewhat more
complicated.  To obtain our characterization we instead use a
classical decomposition theorem due to Eilenberg, Elgot, and
Shepherson \cite{EES69}.  The motivation for studying such sets comes
from the well-known result of Cobham (and its multi-dimensional
generalization due to Semenov) concerning the sets recognizable in
integer bases.

Let $k \geq 2$ be an integer.  A set $X \subseteq \N$ is
\emph{$k$-recognizable} (or \emph{$k$-automatic}) if the language
consisting of the base-$k$ representations of the elements of $X$ is
accepted by a finite automaton.  A celebrated result of Cobham
\cite{Cob69} characterizes the sets that are recognizable in all
integer bases $k \geq 2$.

\begin{theorem}[Cobham]\label{cobham}
Let $k,\ell \geq 2$ be two multiplicatively independent integers and
let $X \subseteq \N$.  The set $X$ is both $k$-recognizable and
$\ell$-recognizable if and only if it is ultimately periodic.
\end{theorem}

Two numbers $k$ and $\ell$ are \emph{multiplicatively independent} if
$k^m = \ell^n$ implies $m=n=0$.  A subset of the integers is
\emph{ultimately periodic} if it is a finite union of arithmetic
progressions.  We say that a set $X \subseteq \N$ is
\emph{$1$-recognizable} if the language $\{a^n : n \in X\}$ consisting
of the unary representations of the elements of $X$ is accepted by a
finite automaton.  It is well-known
\cite[Proposition~V.1.1]{Eilenberg} that a set is
$1$-recognizable if and only if it is ultimately periodic.

Lecomte and Rigo \cite{LR01} introduced the following generalization
of the standard integer base numeration systems.

\begin{definition}
An \emph{abstract numeration system} is a triple $S = (L, \Sigma, <)$
where $L$ is an infinite regular language over a totally ordered
finite alphabet $(\Sigma, <)$.  The map $\rep_S : \N \to L$ is
a bijection mapping $n \in \N$ to the $(n+1)$-th word of $L$
ordered genealogically.  The inverse map is denoted by $\val_S : L \to
\N$.
\end{definition}

Lecomte and Rigo \cite{LR01} proved that any ultimately periodic set is
$S$-recognizable for any abstract numeration system $S$.  Suppose on
the other hand that $X \subseteq \N$ is $S$-recognizable for every
abstract numeration system $S$.  Then in particular, the set $X$ must
be $1$-recognizable, and hence must be ultimately periodic.  We
therefore have the following characterization of the sets that are
recognizable in all abstract numeration systems.

\begin{theorem}[Lecomte and Rigo]\label{lec_rig}
A set $X \subseteq \N$ is $S$-recognizable for all
abstract numeration systems $S$ if and only if it is ultimately
periodic.
\end{theorem}

Rigo and Maes \cite{RM02} considered $S$-recognizability in a
multi-dimensional setting.  This concept was further studied by
Charlier, K\"arki, and Rigo \cite{CKR10}.
For the formal definitions we need to introduce the following
``padding'' function.

\begin{definition}
If $w_1,\ldots, w_d$ are finite words over the alphabet $\Sigma$, the
padding map 
\[
(\cdot)^\# : (\Sigma^*)^d \to ((\Sigma \cup \{\#\})^d)^*
\]
is defined by
\[
(w_1,\ldots,w_d)^\# := (w_1\#^{m-|w_1|},\ldots,w_d\#^{m-|w_d|})
\]
where $m = \max\{|w_1|,\ldots, |w_d|\}$.  Here we write $(ac,bd)$ to
denote the concatenation $(a,b)(c,d)$.

If $R \subseteq (\Sigma^*)^d$, then
\[
R^{\#} = \{(w_1,\ldots,w_d)^\# \colon (w_1,\ldots,w_d) \in R\}.
\]
Note that $R$ is not necessarily a language, whereas $R^{\#}$ is; that
is, the set $R$ consists of $d$-tuples of words over $\Sigma$, whereas
$R^{\#}$ consists of words over the alphabet $(\Sigma \cup \{\#\})^d$.
\end{definition}

\begin{definition}
Let $S = (L, \Sigma, <)$ be an abstract numeration system.  Let $X
\subseteq \N^d$.  The set $X$ is \emph{$S$-recognizable}
(or \emph{$S$-automatic}) if the language $\rep_S(X)^\#$ is regular,
where
\[
\rep_S(X) = \{(\rep_S(n_1),\ldots,\rep_S(n_d)) \colon (n_1,\ldots,n_d)
\in X\}.
\]
Let $k \geq 2$ be an integer.  The set $X$ is \emph{$k$-recognizable} (or
\emph{$k$-automatic}) if it
is $S$-recognizable for the abstract numeration system $S$ built on
the language consisting of the base-$k$ representations of the
elements of $X$.
In particular, the set $X$ is \emph{$1$-recognizable} (or \emph{$1$-automatic}) if it
is $S$-recognizable for the abstract numeration system $S$ built on $a^*$.
\end{definition}

In order to have a multi-dimensional analogue of Cobham's theorem, we
need an analogous notion of ultimate periodicity in the
multi-dimensional setting.  In view of Theorem \ref{cob_sem} below,
the correct generalization turns out to be the following.

\begin{definition}
A set $X \subseteq \N^d$ is \emph{linear} if there exists $v_0,v_1,
\cdots, v_t \in \N^d$ such that
\[
X = \{v_0 + n_1v_1 + n_2v_2 + \cdots + n_tv_t : n_1, \ldots, n_t \in
\N\}.
\]
A set $X \subseteq \N^d$ is \emph{semi-linear} if it is a finite union
of linear sets.
\end{definition}

For more on semi-linear sets see \cite{GS66}.  We can now state the
multi-dimensional version of Cobham's theorem \cite{Sem77}.

\begin{theorem}[Cobham--Semenov]\label{cob_sem}
Let $k,\ell \geq 2$ be two multiplicatively independent integers and
let $X \subseteq \N^d$.  The set $X$ is both $k$-recognizable and
$\ell$-recognizable if and only if it is semi-linear.
 \end{theorem}

In other words, the semi-linear sets are precisely the sets
recognizable in all integer bases $k \geq 2$.  One might therefore
expect that, as in Theorem~\ref{lec_rig}, the semi-linear sets are
recognizable in all abstract numeration systems.  However, this fails
to be the case, as the following example shows.

\begin{example}\label{n_2n}
The semi-linear set $X = \{n(1,2) : n \in \N\} = \{(n,2n) : n \in \N\}$ is
not $1$-recognizable.  Consider the language $\{(a^n\#^n,a^{2n}) : n
\in \N\}$, consisting of the unary representations of the elements of
$X$.  An easy application of the pumping lemma shows that this is not
a regular language.
\end{example}

Observe that in the one-dimensional case, we have the following
equivalences: semi-linear $\Leftrightarrow$ ultimately periodic
$\Leftrightarrow$ $1$-recognizable.  However, Example~\ref{n_2n} shows
that these equivalences no longer hold in the multi-dimensional
setting.  In order to get a multi-dimensional analogue of
Theorem~\ref{lec_rig}, we must consider the class of $1$-recognizable
sets, which form a proper subclass of the class of semi-linear sets.

Another well-studied subclass of the class of semi-linear sets is the
class of recognizable sets.  A subset $X$ of $\N^d$ is
\emph{recognizable} if there exists a finite monoid $M$, a monoid
homomorphism $\varphi : \N^d \to M$, and a subset $B \subseteq M$ such
that $X = \varphi^{-1}(B)$.  When $d=1$, we have again the following equivalences: recognizable $\Leftrightarrow$ ultimately periodic
$\Leftrightarrow$ $1$-recognizable.  However, for $d>1$ these
equivalences no longer hold.  An unpublished result of Mezei (see
\cite[Proposition~III.12.2]{Eilenberg}) demonstrates that the
recognizable subsets of $\N^2$ are precisely finite unions of sets of
the form $Y \times Z$, where $Y$ and $Z$ are ultimately periodic
subsets of $\N$.  In particular, the \emph{diagonal set} $D = \{(n,n)
: n \in \N\}$ is not recognizable \cite[Exercise~III.12.7]{Eilenberg}.
However, the set $D$ is clearly a $1$-recognizable subset of $\N^2$.
So we see that for $d>1$, the class of $1$-recognizable sets
corresponds neither to the class of semi-linear sets, nor to the class
of recognizable sets.   For further information on recognizable sets,
see \cite{BHMV94}.

Our main result is the following:

\begin{theorem}\label{main}
Let $X \subseteq \N^d$.  Then $X$ is $S$-recognizable for all
abstract numeration systems $S$ if and only if $X$ is $1$-recognizable.
\end{theorem}

To illustrate this theorem, we give the following example.

\begin{example}\label{ex_2D}
Let
\[
X = \{(2n,3m+1) : n,m \in\N \text{ and } 2n \geq 3m+1\} \cup
\{(n,2m) : n,m  \in\N \text{ and } n < 2m\}.
\]
\begin{figure}[htbp]
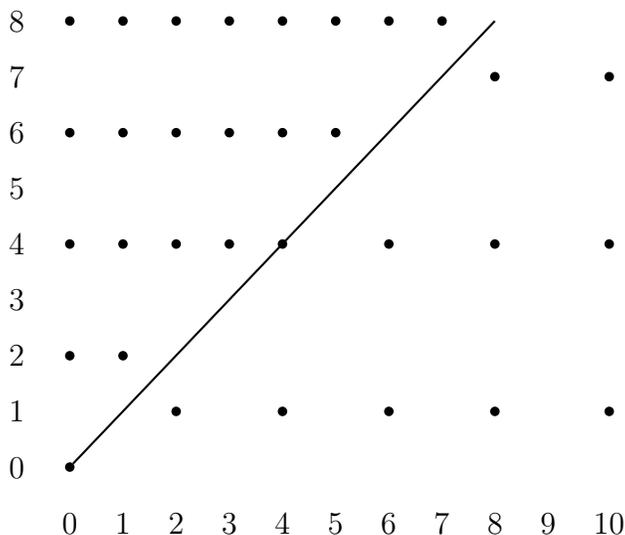

\begin{center}
\begin{psmatrix}[mnode=dot,colsep=0.5cm,rowsep=0.25cm]
[mnode=r]8&\ &\ &\ &\ &\ &\ &\ &\ & [mnode=r]\ \\
[mnode=r]7&   &  &   &   &   &  &   &  &\ & & \ \\
[mnode=r]6&\ &\ &\ &\ &\ &\ \\
[mnode=r]5 \\
[mnode=r]4&\ &\ &\ &\ &\ &  &\ &  &\ &  &\ \\
[mnode=r]3 \\
[mnode=r]2&\ &\ \\
[mnode=r]1&  &  &\ &  &\ &  &\ &  &\  &  &\ \\
[mnode=r]0&\ \\
&[mnode=r]0&[mnode=r]1&[mnode=r]2&[mnode=r]3&[mnode=r]4&[mnode=r]5&[mnode=r]6&[mnode=r]7&[mnode=r]8&[mnode=r]9&[mnode=r]10
\end{psmatrix}
\ncline{9,2}{1,10}
\end{center}
\caption{The set $X$ of Example~\ref{ex_2D}}
\end{figure}
It is clear that $X$ is $1$-recognizable.  Let $S = (L,\Sigma,<)$ be
an abstract numeration system.  By Theorem~\ref{lec_rig}, the sets $\{2n
: n \in \N\}$ and $\{3m+1 : m \in \N\}$ are both $S$-recognizable, and
so the set $\{(2n,3m+1) : n,m \in\N\}$ is also $S$-recognizable.  In
other words, the set $\{(\rep_S(2n), \rep_S(3m+1))^\# : n,m \in \N\}$
is accepted by a finite automaton.  Furthermore, the
set $\{(x,y)^\# : x,y \in L \text{ and } x \geq y\}$ is also accepted
by a finite automaton, and so by taking the product of these two
automata we obtain an automaton accepting
\[
\{(\rep_S(2n), \rep_S(3m+1))^\# : n,m \in \N \text{ and } 2n \geq
3m+1\}.
\]
In the same way we can construct an automaton to accept the set
\[
\{(\rep_S(n),\rep_S(2m))^\# : n,m  \in\N \text{ and } n < 2m\}.
\] 
Since the union of two regular languages is regular, we see that
$X$ is $S$-recognizable.
\end{example}

\section{Proof of our main result}
In order to obtain our main result, we will need a classical result of
Eilenberg, Elgot, and Shepherdson \cite[Theorem~11.1]{EES69} (see also
\cite[Theorem~C.1.1]{Rub04}).  We first need the following definition.

\begin{definition}
Let $A$ be a non-empty subset of $\{1,\ldots,d\}$.  Define the
subalphabet 
\[
\Sigma_A = \{x \in (\Sigma \cup \{\#\})^d : \text{the $i$-th component
  of $x$ is \# exactly when } i \notin A\}.
\]
\end{definition}

\begin{example}
Let $\Sigma = \{a\}$ and $d=4$.  If $A = \{1,2,3,4\}$, then
$\Sigma_A = \{(a,a,a,a)\}$.  If $A = \{2,3\}$, then
$\Sigma_A = \{(\#,a,a,\#)\}$.  If $A = \{3\}$, then
$\Sigma_A = \{(\#,\#,a,\#)\}$.
\end{example}

\begin{theorem}[Decomposition \cite{EES69}]
\label{decomp}
Let $R \subseteq (\Sigma^*)^d$.  The language $R^{\#} \subseteq
((\Sigma \cup \{\#\})^d)^*$ is regular if and only if it is a finite
union of languages of the form
\[
R_0 \cdots R_t, \quad t \in \N,
\]
where each factor $R_i \subseteq (\Sigma_{A_i})^*$ is regular and
$A_t \subseteq \cdots \subseteq A_0 \subseteq \{1,\ldots,d\}$.
\end{theorem}

\begin{remark}
Theorem~\ref{decomp} does not hold if $R^\#$ is replaced by an
arbitrary language over $(\Sigma \cup \{\#\})^d$.  It is only valid
due to the definition of the map $(\cdot)^\#$.
\end{remark}

\begin{example}
Let $R = \{(a^{5n},a^{6m}): n,m \in \N\}$.  Then $R^{\#}$ is regular,
since one can easily construct an automaton that simultaneously checks
that the length of the first component of its input is a multiple of $5$
and that the length of the second component is a multiple of $6$.
Moreover, we have
\[
R^{\#} = \bigcup_{\ell=0}^5 (a^{30}, a^{30})^*(a^{5\ell}\#^\ell, a^{6\ell})(\#^6,
a^6)^* \quad\cup\quad \bigcup_{\ell=0}^4 (a^{30}, a^{30})^*(a^{5(\ell+1)},
a^{6\ell}\#^{5-\ell})(a^5, \#^5)^*.
\]
Observe that each of the languages appearing in the unions above are
products of the form described in Theorem~\ref{decomp}.
\end{example}

\begin{lemma}\label{ugly_lemma}
Let $X \subseteq \N^d$.  Then $X$ is $1$-recognizable if and
only if $X$ is a finite union of sets of the form
\begin{multline}\label{ugly_sum}
\Bigg\lbrace
\sum_{\ell=0}^t (c_\ell(n_{\ell,1},\ldots,n_{\ell,d}) +
  (b_{\ell,1},\ldots,b_{\ell,d})) \colon (\forall \ell)(\forall i)\;
  n_{\ell, i} \in \N \text{ and } \\ (\forall \ell)(\forall i)\;
  (i \notin A_\ell \Rightarrow n_{\ell,i} = 0) \text{ and } (\forall
  \ell)(\forall i)(\forall j) (i,j \in A_\ell \Rightarrow n_{\ell,i} = n_{\ell,j})
\Bigg\rbrace
\end{multline}
where
\begin{itemize}
\item $t \in \N$,
\item $A_t \subseteq \cdots \subseteq A_0 \subseteq \{1,\ldots,d\}$,
\item $c_0, \ldots, c_t \in \N$,
\item $(\forall \ell)(\forall i)\; b_{\ell, i} \in \N$,
\item $(\forall \ell)(\forall i)\; (i \notin A_\ell \Rightarrow b_{\ell,i} = 0)$, and
\item $(\forall \ell)(\forall i)(\forall j)\; (i,j \in A_\ell \Rightarrow b_{\ell,i} =
  b_{\ell,j})$.
\end{itemize}
\end{lemma}

\begin{proof}
Let $\Sigma = \{a\}$ and let $S = (\Sigma^*,\Sigma,<)$.  We define
\[
R := \rep_S(X)=\{(a^{n_1},\ldots,a^{n_d}) \colon (n_1,\ldots,n_d) \in X\}.
\]
The set $X$ is $1$-recognizable if and only if the language $R^{\#}$
is regular.  By Theorem~\ref{decomp}, the language $R^\#$ is regular
if and only if it is a finite union of languages of
the form
\[
R_0 \cdots R_t, \quad t \in \N,
\]
where each factor $R_\ell \subseteq (\Sigma_{A_\ell})^*$ is regular and
$A_t \subseteq \cdots \subseteq A_0 \subseteq \{1,\ldots,d\}$.  Since
$|\Sigma|=1$, we have $|\Sigma_{A_\ell}| = 1$.  Let $\Sigma_{A_\ell} =
\{x\}$.  It is well-known \cite[Proposition~V.1.1]{Eilenberg} that
$R_\ell$ is a finite union of languages of the
form $\{x^{pi+q} : i \in \N\}$, where $p,q \in \N$.  Without loss of
generality we can assume that $R_\ell$ is exactly of this form.
Hence, the language $R_\ell$ consists of the representations of a set
of the form
\[
\{c_\ell(n_{\ell,1},\ldots,n_{\ell,d}) +
(b_{\ell,1},\ldots,b_{\ell,d}) : (\forall i) (n_{\ell,i} \in \N)\}.
\]
The conditions $A_t \subseteq \cdots \subseteq A_0 \subseteq
\{1,\ldots,d\}$ impose the restrictions on the $n_{\ell,i}$'s and
the constants  $b_{\ell,i}$ in the statement of the lemma.
The concatenation of the $R_\ell$'s gives the sum described above.
\end{proof}

\begin{example}\label{ex_part1}
Let $X = \{(5n, 5n+4m+6\ell+1, 5n+4m+6\ell+3, 5n) : n,m,\ell \in\N\}$.
The unary representation of $X$ is
\[
R^{\#} = ((a,a,a,a)^5)^*((\#,a,a,\#)^4)^*((\#,a,a,\#)^6)^*(\#,a,a,\#)(\#,\#,a,\#)^2.
\]
Since $R^{\#}$ is regular the set $X$ is $1$-recognizable.  The set
$X$ can be written as
\[
X = \{5(n,n,n,n) + 4(0,m,m,0) + 6(0,\ell,\ell,0) + (0,1,1,0) + (0,0,2,0) :
n,m,\ell \in \N\},
\]
which is an expression of the form \eqref{ugly_sum} where
$t=3$; $A_0 = \{1,2,3,4\}$, $A_1 = A_2 = \{2,3\}$, $A_3 = \{3\}$;
$c_0=5, c_1=4, c_2=6, c_3=0$; and $b_{0,i} = b_{1,i} = 0$ for all $i$,
$(b_{2,1},b_{2,2}, b_{2,3},b_{2,4}) = (0,1,1,0)$, $(b_{3,1},b_{3,2}, b_{3,3},b_{3,4}) =
(0,0,2,0)$.

Furthermore, we have a factorization of $R^{\#}$ as given in
Theorem~\ref{decomp}: $R^{\#} = R_0R_1R_2R_3$, where $R_0 =
((a,a,a,a)^5)^*$, $R_1 = ((\#,a,a,\#)^4)^*$, $R_2 =
((\#,a,a,\#)^6)^*(\#,a,a,\#)$, and $R_3 = (\#,\#,a,\#)^2$, with the same
$A_\ell$'s as those defined above.  The term $5(n,n,n,n)$
corresponds to $R_0$, the term $4(0,m,m,0)$ corresponds to $R_1$, the
term $6(0,\ell,\ell,0) + (0,1,1,0)$ corresponds to $R_2$, and the term
$(0,0,2,0)$ corresponds to $R_3$.
\end{example}

In the sequel we write ${\bf e}_i$ to denote the element of $\N^d$
that contains a 1 in its $i$-th component and 0's in all others.

\begin{lemma}\label{pairwise_comparison}
A set $X \subseteq \N^d$ of the form \eqref{ugly_sum} can be written as a union $A \cup
B$, where $A$ is made up of finite unions and intersections of sets having one of the forms
\eqref{only_j}--\eqref{j_and_k_finite} below and $B$ is a finite
intersection of sets of the form \eqref{only_j} or \eqref{j_and_k} below:
\begin{equation}\label{only_j}
\left\{ \sum_{\substack{i=1\\i \neq j}}^d n_i {\bf e}_i + (rn_j +
  s){\bf e}_j : n_1, \ldots, n_d \in \N, n_j \ge N \right\}
\end{equation}
where $1 \leq j \leq d$,  and $r,s,N \in \N$;
\begin{equation}\label{j_and_k}
\left\{ \sum_{\substack{i=1\\i \neq j}}^d n_i {\bf e}_i + (n_k + rn_j +
  s){\bf e}_j : n_1, \ldots, n_d \in \N , n_j \ge N \right\}
\end{equation}
where $1 \leq j,k \leq d$,  $j \neq k$, and $r,s,N \in \N$;
\begin{equation}\label{only_j_finite}
\left\{ \sum_{\substack{i=1\\i \neq j}}^d n_i {\bf e}_i + (rn_j +
  s){\bf e}_j : n_1, \ldots, n_d \in \N, n_j  \in C \right\}
\end{equation}
where $1 \leq j \leq d$, $r,s \in \N$, and $C \subseteq \N$
is a finite set; or
\begin{equation}\label{j_and_k_finite}
\left\{ \sum_{\substack{i=1\\i \neq j}}^d n_i {\bf e}_i + (n_k + rn_j +
  s){\bf e}_j : n_1, \ldots, n_d \in \N , n_j \in C \right\}
\end{equation}
where $1 \leq j,k \leq d$, $j \neq k$, and $r,s \in \N$,
and $C \subseteq \N$ is a finite set.
\end{lemma}

\begin{proof}
Let $X$ be a set of the form \eqref{ugly_sum} where $t$, the
$A_\ell$'s, the $c_\ell$'s, and the $b_{\ell,i}$'s are fixed and
satisfy the conditions listed in Lemma~\ref{ugly_lemma}.
We will write $X = A \cup B$, where
\[
B = \bigcap_{j=1}^d Y_j,
\]
where each $Y_j$ is either of the form \eqref{only_j} or \eqref{j_and_k}, and 
$A$ is made up of finite unions and intersections of sets of the forms \eqref{only_j}--\eqref{j_and_k_finite}.

First observe that if $j \in \{1,\ldots,d\} \setminus A_0$ the set $X$
contains only vectors whose $j$-th component is always $0$.  For each
such $j$, we define
\[
Y_j = \left\{ \sum_{\substack{i=1\\i \neq j}}^d n_i {\bf e}_i + 0{\bf e}_j : n_1, \ldots, n_d \in \N \right\},
\]
which is of the form \eqref{only_j}.

First consider the case where $A_0=\cdots=A_t$. Define $j_1<\cdots<j_{|A_0|}$ to be the elements of $A_0$. Define 
\[
Y_{j_1} = \left\{\sum_{\substack{i=1\\i \neq j_1}}^d n_i {\bf e}_i +
  (rn_{j_1} + s){\bf e}_{j_1} : n_1, \ldots, n_d \in \N,
  n_{j_1} \ge N \right\},
\]
where $r = \gcd(c_0,\ldots,c_t)$, $s=\sum_{\ell=0}^t b_{\ell,j_1}$,
and $N-1$ is the largest
integer $n$ such that $rn$ cannot be written as a nonnegative integer
linear combination of $c_0,\ldots,c_t$ (note that $N$ exists
and is finite \cite[Theorem~1.0.1]{RA05}).  Note that $Y_{j_1}$ is of the form
\eqref{only_j}.

Define
\[
Y_{j_1}' = \left\{\sum_{\substack{i=1\\i \neq j_1}}^d n_i {\bf e}_i +
  (rn_{j_1} + s){\bf e}_{j_1} : n_1, \ldots, n_d \in \N,
  n_{j_1} \in C \right\},
\]
where $C$ is the set of all nonnegative integers $n < N$
such that $rn$ can be written as a nonnegative integer
linear combination of $c_0,\ldots,c_t$.  Note that $Y_{j_1}'$ is of
the form \eqref{only_j_finite}.

 For $k\in\{2,\ldots,|A_0|\}$, define 
\[
Y_{j_k} = \left\{ \sum_{\substack{i=1\\i \neq j_k}}^d n_i {\bf e}_i +
  n_{j_{k-1}}{\bf e}_{j_k} : n_1, \ldots, n_d \in \N \right\},
\]
which is of the form \eqref{j_and_k}.

The set $X$ can be written as the union $ A\cup B$ where 
\[
B=\bigcap_{j\in\{1,\ldots,d\}\setminus A_0} Y_j \quad\cap\quad \bigcap_{k\in\{1,\ldots,|A_0|\}} Y_{j_k}
\]
and 
\[
A=\bigcap_{j\in\{1\ldots,d\}\setminus A_0} Y_j \quad\cap\quad
\bigcap_{k\in\{2\ldots,|A_0|\}} Y_{j_k} \quad\cap\quad Y_{j_1}'.
\]

Now consider the case where there is at least one index $\ell$ such
that $A_{\ell} \setminus A_{\ell+1}\neq \emptyset$. Define
$\ell_1<\cdots<\ell_{t'}$ to be the indices of the sets $A_\ell$
satisfying $A_{\ell_k} \setminus A_{\ell_k+1}\neq \emptyset$ for each
$k\in\{1,\ldots,t'\}$. We clearly have $1\le t'\le t$ and $0\le\ell_{t'}<t$.

Define $d_1=| A_{\ell_1} \setminus A_{\ell_1+1}|$ and
$j_{1,1}<\cdots<j_{1,d_1}$ to be the elements of $A_{\ell_1} \setminus
A_{\ell_1+1}$. Define
\[
Y_{j_{1,1}} = \left\{\sum_{\substack{i=1\\i \neq j_{1,1}}}^d n_i {\bf e}_i +
  (r_{1}n_{j_{1,1}} + s_{1}){\bf e}_{j_{1,1}} : n_1, \ldots, n_d \in \N,
  n_{j_{1,1}} \ge N_{1} \right\},
\]
where $r_{1} = \gcd(c_0,\ldots,c_{\ell_1})$, $s_{1} =
\sum_{\ell=0}^{\ell_1} b_{\ell,j_{1,1}}$, and 
$N_1-1$ is the largest integer $n$ such that $r_{1}n$ cannot be written
as a nonnegative integer linear combination of
$c_0,\ldots,c_{\ell_1}$.  Note that $Y_{j_{1,1}}$ is of the form
\eqref{only_j}.

Define
\[
Y_{j_{1,1}}' = \left\{\sum_{\substack{i=1\\i \neq j_{1,1}}}^d n_i {\bf e}_i +
  (r_1n_{j_{1,1}} + s_1){\bf e}_{j_{1,1}} : n_1, \ldots, n_d \in \N,
  n_{j_{1,1}} \in C_1 \right\},
\]
where $C_1$ is the set of all nonnegative integers $n < N_1$ such that
$r_1n$ can be written as a nonnegative integer
linear combination of $c_0,\ldots,c_{\ell_1}$.  Note that $Y_{j_{1,1}}'$ is of
the form \eqref{only_j_finite}.

 For $k\in\{2,\ldots,d_1\}$, define 
\[
Y_{j_{1,k}} = \left\{ \sum_{\substack{i=1\\i \neq j_{1,k}}}^d n_i {\bf e}_i +
  n_{j_{1,k-1}}{\bf e}_{j_{1,k}} : n_1, \ldots, n_d \in \N
\right\},
\]
which is of the form \eqref{j_and_k}.

Define $d_2=| A_{\ell_2} \setminus
A_{\ell_2+1}|$ and $j_{2,1}<\cdots<j_{2,d_2}$ to be the elements of $A_{\ell_2} \setminus
A_{\ell_2+1}$. 
Define 
\[
Y_{j_{2,1}} = \left\{ \sum_{\substack{i=1\\i \neq j_{2,1}}}^d n_i {\bf e}_i +
  (n_{j_{1,1}} + r_2n_{j_{2,1}} + s_2){\bf e}_{j_{2,1}} : n_1, \ldots, n_d \in \N
  , n_{j_{2,1}} \ge N_2 \right\},
\]
where $r_{2} = \gcd(c_{\ell_1+1},\ldots,c_{\ell_2})$, $s_{2} =
\sum_{\ell=\ell_1+1}^{\ell_2} b_{\ell,j_{2,1}}$, and $N_2-1$ is the
largest integer $n$ such that $r_{2}n$ cannot be written as a
nonnegative integer linear combination of
$c_{\ell_1+1},\ldots,c_{\ell_2}$.  Note that $Y_{j_{2,1}}$ is of the
form \eqref{j_and_k}.

Define 
\[
Y_{j_{2,1}}' = \left\{ \sum_{\substack{i=1\\i \neq j_{2,1}}}^d n_i {\bf e}_i +
  (n_{j_{1,1}} + r_2n_{j_{2,1}} + s_2){\bf e}_{j_{2,1}} : n_1, \ldots, n_d \in \N
  , n_{j_{2,1}} \in C_2 \right\},
\]
where $C_2$ is the set of all nonnegative integers $n < N_2$ such that
$r_2n$ can be written as a nonnegative integer linear
combination of $c_{\ell_1+1},\ldots,c_{\ell_2}$.  Note that
$Y_{j_{2,1}}'$ is of the form \eqref{j_and_k_finite}.

 For $k\in\{2,\ldots,d_2\}$, define 
\[
Y_{j_{2,k}} = \left\{ \sum_{\substack{i=1\\i \neq j_{2,k}}}^d n_i {\bf e}_i +
  n_{j_{2,k-1}}{\bf e}_{j_{2,k}} : n_1, \ldots, n_d \in \N
\right\},
\]
which is of the form \eqref{j_and_k}.

We continue in this manner to define $d_p$, $Y_{j_{p,k}}$, and
$Y_{j_{p,1}}'$ for all $p\in\{1,\ldots,t'\}$ and
$k\in\{1,\ldots,d_p\}$.  Finally observe that we have $A_{\ell_{t'}}\setminus
A_{\ell_{t'}+1}\neq \emptyset$ and $ A_{\ell_{t'}+1}=\cdots=A_t$. Define
$d_{t'+1}=|A_t|$ and $j_{t'+1,1}<\cdots<j_{t'+1,d_{t'+1}}$ to be
the elements of $A_t$. Define
\[
Y_{j_{t'+1,1}} = \left\{ \sum_{\substack{i=1\\i \neq j_{t'+1,1}}}^d n_i {\bf e}_i +
 (n_{j_{t',1}}+r_{t'+1} n_{j_{t'+1,1}}+s_{t'+1}){\bf e}_{j_{t'+1,1}} : n_1, \ldots, n_d \in \N,
 n_{j_{t'+1,1}} \ge N_{t'+1} \right\},
\]
where $r_{t'+1} = \gcd(c_{\ell_{t'}+1},\ldots,c_{t})$, $s_{t'+1} =
\sum_{\ell=\ell_{t'}+1}^{\ell_t} b_{\ell,j_{t'+1,1}}$, and $N_{t'+1}-1$
is the largest integer $n$ such that $r_{t'+1}n$ cannot be written as
a nonnegative integer linear combination of
$c_{\ell_{t'}+1},\ldots,c_{t}$.  Again note that $Y_{j_{t'+1,1}}$ is
of the form \eqref{j_and_k}.

Define
\[
Y_{j_{t'+1,1}}' = \left\{ \sum_{\substack{i=1\\i \neq j_{t'+1,1}}}^d
  n_i {\bf e}_i + (n_{j_{t',1}}+r_{t'+1} n_{j_{t'+1,1}}+s_{t'+1}){\bf
    e}_{j_{t'+1,1}} : n_1, \ldots, n_d \in \N, n_{j_{t'+1,1}} \in
  C_{t'+1} \right\},
\]
where $C_{t'+1}$ is the set of all nonnegative integers $n < N_{t'+1}$
such that $r_{t'+1}n$ can be written as a nonnegative integer linear
combination of $c_{\ell_{t'}+1},\ldots,c_{t}$.  Note that
$Y_{j_{t'+1,1}}'$ is of the form \eqref{j_and_k_finite}.

 For $k\in\{2,\ldots,d_{t'+1}\}$, define 
\[
Y_{j_{t'+1,k}} = \left\{ \sum_{\substack{i=1\\i \neq j_{t'+1,k}}}^d n_i {\bf e}_i +
  n_{j_{t'+1,k-1}}{\bf e}_{j_{t'+1,k}} : n_1, \ldots, n_d \in \N
\right\},
\]
which is of the form \eqref{j_and_k}.

The set $X$ can be written as the union $ A\cup B$ where
\[
B=\bigcap_{j\in\{1,\ldots,d\}\setminus A_0} Y_j \quad\cap\quad
\bigcap_{\substack{p\in\{1,\ldots,t'+1\}\\k\in\{1,\ldots,d_p\}}}
Y_{j_{p,k}}
\]
and
\[
A=  \bigcap_{j\in\{1,\ldots,d\}\setminus A_0} Y_j \quad\cap\quad
 \bigcup_{p\in\{1,\ldots,t'+1\}} \left( Y_{j_{p,1}}' \quad\cap\quad \bigcap_{\substack{q\in\{1,\ldots,t'+1\}\setminus\{p\}\\k\in\{1,\ldots,d_q\}}} Y_{j_{q,k}}  \quad\cap\quad \bigcap_{k\in\{2,\ldots,d_p\}} Y_{j_{p,k}}\right).
\]
\end{proof}

\begin{example}\label{ex_part2}
We continue Example~\ref{ex_part1}.  We will write $X = A \cup B$ as in
Lemma~\ref{pairwise_comparison}.  The $A_\ell$'s are not all the same,
so we can define $t' = 2$, $\ell_1 = 0 < \ell_2 = 2$ as in the proof
of Lemma~\ref{pairwise_comparison}.

We have $d_1 = |A_0 \setminus
A_1| = 2$, $j_{1,1} = 1$ and $j_{1,2}=4$.  We also have $r_1 = \gcd(c_0) = \gcd(5) = 5$
and $s_1 = 0$, and hence $N_1 = 0$.  Therefore,
\[
Y_1 = \{n_2{\bf e_2} + n_3{\bf e_3} + n_4{\bf e_4} + (5n_1+0){\bf e_1} : n_1,n_2,n_3,n_4
\in \N, n_1\ge 0\},
\]
\[
Y_1' = \{n_2{\bf e_2} + n_3{\bf e_3} + n_4{\bf e_4} + (5n_1+0){\bf e_1} : n_2,n_3,n_4 \in
\N, n_1 \in C_1\} = \emptyset,
\]
since $C_1 = \emptyset$, and
\[
Y_4 = \{n_1{\bf e_1}+ n_2{\bf e_2}+ n_3{\bf e_3}+ n_1{\bf e_4} : n_1,n_2,n_3 \in \N\}.
\]

Next we have $d_2 = |A_2  \setminus
A_3| = 1$ and $j_{2,1} = 2$.  We also have $r_2 = \gcd(c_1,c_2) = \gcd(4,6) = 2$
and $s_2 = b_{1,2} + b_{2,2} = 0+1 = 1$, and hence $N_2 = 2$.  Therefore,
\[
Y_2 = \{n_1{\bf e_1} + n_3{\bf e_3} + n_4{\bf e_4} + (n_1+2n_2+1){\bf e_2} : n_1,n_2,n_3
\in \N, n_2 \ge 2\},
\]
and
\begin{eqnarray*}
Y_2' & = & \{n_1{\bf e_1} + n_3{\bf e_3} + n_4{\bf e_4} + (n_1+2n_2+1){\bf e_2} : n_1,n_3
\in \N, n_2 \in C_2 \} \\
& = & \{n_1{\bf e_1} + n_3{\bf e_3} + n_4{\bf e_4} + (n_1+1){\bf e_2} : n_1,n_3
\in \N \},
\end{eqnarray*}
since $C_2 = \{0\}$.

Finally, we have $d_3 = |A_3| = 1$ and $j_{3,1} = 3$.  We also have
$r_3 = \gcd(c_3) = \gcd(0) = 0$ and $s_3 = b_{3,3} = 2$, and hence
$N_3 = 0$.  Therefore,
\[
Y_3 = \{n_1{\bf e_1} + n_2{\bf e_2} + n_4{\bf e_4} + (n_2+0n_3+2){\bf e_3} : n_1,n_2,n_3
\in \N, n_3 \ge 0\},
\]
and
\[
Y_3' = \{n_1{\bf e_1} + n_2{\bf e_2} + n_4{\bf e_4} + (n_2+0n_3+2){\bf e_3} : n_1,n_2
\in \N, n_3 \in C_3\} = \emptyset,
\]
since $C_3 = \emptyset$.

Hence $A = Y_1\cap Y_2' \cap Y_3\cap Y_4$ and $B = Y_1 \cap Y_2 \cap Y_3 \cap Y_4$.
\end{example}

\begin{lemma}\label{n_plus_k}
Let $k \in \N$ and let $S$ be an abstract numeration system.  The set
$X = \{(n,n+k) : n \in \N\}$ is $S$-recognizable.
\end{lemma}

\begin{proof}
Let $R = \rep_S(X)$.  To show that $X$ is $S$-recognizable we must
show that $R^\#$ is a regular language.  Consider first the set $Y =
\{(\rep_S(n),\rep_S(n+1)) : n \in \N\}$.  If we interpret $Y$ as the
function mapping $\rep_S(n)$ to $\rep_S(n+1)$, then $Y$ is the
so-called \emph{successor function} (see \cite{AS10} or \cite{LR01}
for more on the successor function).  From
\cite[Proposition~3]{BFRS07} (see also
\cite[Proposition~2.6.7]{CANT10}), we have that $Y$ is a
\emph{synchronous relation}.  In \cite{CANT10} synchronous relations
are defined in terms of letter-to-letter transducers, but this
definition is equivalent to the fact that the language $Y^\#$ is
accepted by a finite automaton.  Moreover, from \cite{FS93} (see
also \cite[Theorem~2.6.6]{CANT10}), we have that the composition of
synchronous relations is again a synchronous relation.  Hence $R$,
which is the $k$-fold composition of $Y$ with itself, is a
synchronous relation.  We conclude that $R^\#$ is a regular
language, as required.
\end{proof}

\begin{lemma}\label{S-recog}
A set $X \subseteq \N^d$ having one of the forms
\eqref{only_j}--\eqref{j_and_k_finite} defined in
Lemma~\ref{pairwise_comparison} is $S$-recognizable for any abstract
numeration system $S$.
\end{lemma}

\begin{proof}
We will give the proof for the cases where $X$ is either of the form
\eqref{only_j} or \eqref{j_and_k} (the other two cases are similar).

Let $S = (L,\Sigma,<)$ be an abstract numeration system and let
$\mathcal{T}$ be a finite automaton accepting $L$.  Let $R =
\rep_S(X)$.  We will show that $R^\#$ is regular.  That is, we will
define a (nondeterministic) finite automaton $\mathcal{M}$ that
accepts $R^\#$.  Let $(w_1,\ldots,w_d)^\#$ be an
arbitrary input to the automaton $\mathcal{M}$.

Suppose that $X$ is of the form \eqref{only_j}.  That is,
\[
X = \left\{ \sum_{\substack{i=1\\i \neq j}}^d n_i {\bf e}_i + (rn_j +
  s){\bf e}_j : n_1, \ldots, n_d \in \N, n_j \geq N \right\},
\]
where $1 \leq j \leq d$, and $r,s,N \in \N$.  Suppose first that
$r=0$.  In this case, the automaton $\mathcal{M}$ simulates
$\mathcal{T}$ on $w_1,\ldots,w_{j-1},w_{j+1},\ldots,w_d$.
The automaton $\mathcal{M}$ accepts its input if and only
if $\mathcal{T}$ accepts $w_1,\ldots,w_{j-1},w_{j+1},\ldots,w_d$ and
$w_j = \rep_S(s)$.

Now suppose that $r>0$.  By increasing the value of $N$, we may,
without loss of generality, assume that $s<r$.  By
\cite[Theorem~4]{LR01} (see also \cite[Theorem~3.3.1]{CANT10b}), the
language $\{\rep_S(rn_j+s) : n_j \in \N\}$ is regular, and hence the
language $\{\rep_S(rn_j+s) : n_j \geq N\}$ is also regular (since it
differs from the former only by a finite set).  Let $\mathcal{T}'$ be
an automaton accepting $\{\rep_S(rn_j+s) : n_j \geq N\}$.  As before, the
automaton $\mathcal{M}$ simulates $\mathcal{T}$ on
$w_1,\ldots,w_{j-1},w_{j+1},\ldots,w_d$, but now also simulates
$\mathcal{T}'$ on $w_j$.  The automaton $\mathcal{M}$ accepts its
input if and only if $\mathcal{T}$ accepts
$w_1,\ldots,w_{j-1},w_{j+1},\ldots,w_d$ and $\mathcal{T}'$ accepts
$w_j$.

Next suppose that $X$ is of the form \eqref{j_and_k}.  That is,
\[
\left\{ \sum_{\substack{i=1\\i \neq j}}^d n_i {\bf e}_i + (n_k + rn_j +
  s){\bf e}_j : n_1, \ldots, n_d \in \N , n_j \geq N \right\},
\]
where $1 \leq j,k \leq d$, $j \neq k$, and $r,s,N \in \N$.  Again,
suppose first that $r=0$.  By Lemma~\ref{n_plus_k}, the language
$\{(\rep_S(n_k),\rep_S(n_k+s))^\# : n_k \in \N\}$ is regular.  Let
$\mathcal{T}''$ be a finite automaton accepting this language.  The
automaton $\mathcal{M}$ simulates $\mathcal{T}$ on each of the words
in $\{w_1,\ldots,w_d\} \setminus \{w_j,w_k\}$.  Simultaneously, the
automaton $\mathcal{M}$ simulates $\mathcal{T}''$ on the pair $(w_k,
w_j)^\#$.  The automaton $\mathcal{M}$ accepts its input if and only
if $\mathcal{T}$ accepts $\{w_1,\ldots,w_d\} \setminus \{w_j,w_k\}$
and $\mathcal{T}''$ accepts $(w_k, w_j)^\#$.

Now suppose that $r>0$.  Again, without loss of generality, we may
assume that $s<r$.  Using the same ideas as in the proof of
\cite[Theorem~3.3.1]{CANT10b}, it is not hard to see that the language
\[
\{(\rep_S(m),\rep_S(n))^\# : m,n \in \N \text{ and } (n-m) \equiv s \pmod r\}
\]
is regular.  Let $\mathcal{Z}$ be an automaton accepting this language.  Let
$\mathcal{Z}'$ be an automaton accepting the language $\{(\rep_S(n_k),\rep_S(n_k+rN+s))^\#
: n_k \in \N\}$ (since $rN+s$ is a constant, we may apply
Lemma~\ref{n_plus_k}).

The automaton $\mathcal{M}$ simulates $\mathcal{T}$ on each of the
words in $\{w_1,\ldots,w_d\} \setminus \{w_j,w_k\}$.  Simultaneously,
the automaton $\mathcal{M}$ simulates $\mathcal{Z}$ on the pair
$(w_k,w_j)^\#$.

The automaton $\mathcal{M}$ also nondeterministically ``guesses'' a
word $v = b_1 \cdots b_{|v|}$ and simulates $\mathcal{Z}'$ on the pair
$(w_k,v)^\#$.  This ``guess'' works as follows.  Let $w_k = a_1 \cdots
a_{|w_k|}$, where each $a_i \in \Sigma$.  For each $i =
1,\ldots,|w_k|$, we simulate $\mathcal{Z}'$ by nondeterministically
choosing to follow one of the transitions of $\mathcal{Z}'$ labeled
$(a_i,b_i)$, where $b_i \in \Sigma$; and for $i > |w_k|$ (i.e., $w_k$
has been completely read), the simulation may make a nondeterministic
choice among transitions of the form $(\#,b_i)$, where $b_i \in
\Sigma$.  This nondeterministic choice of $b_i$ at each step of the
simulation is what defines the ``guessed'' word $v$.  Note that if
$\mathcal{Z}'$ accepts $(w_k,v)^\#$, then $\val_s(v) =
\val_S(w_k)+rN+s$.  As this nondeterministic simulation is performed,
the automaton $\mathcal{M}$ also simultaneously verifies that $w_j$ is
greater than or equal to (in the radix order) the guessed word $v$.

The automaton $\mathcal{M}$ accepts its input if and only if
\begin{itemize}
\item $\mathcal{T}$ accepts each of the
words in $\{w_1,\ldots,w_d\} \setminus \{w_j,w_k\}$,

\item $\mathcal{Z}$ accepts $(w_k,w_j)^\#$,

\item $\mathcal{Z}'$ accepts $(w_k,v)^\#$ for some guessed word $v$ as described
  above, and

\item $w_j$ is greater than or equal to $v$ in the radix order.
\end{itemize}
The last three of these conditions guarantee that $\val_S(w_j) =
\val_S(w_k) + rn_j + s$ for some $n_j \geq N$.

This completes the proof for the cases where $X$ is either of the form
\eqref{only_j} or \eqref{j_and_k}.  As previously stated, we omit the
details for the other two cases since they are similar.
\end{proof}

We are ready for the proof of Theorem~\ref{main}.

\begin{proof}[Proof of Theorem~\ref{main}]
One direction is clear: if $X$ is $S$-recognizable for all abstract
numeration systems $S$, then it is certainly $1$-recognizable.

To prove the other direction, suppose that $X$ is $1$-recognizable.
The result now follows from Lemmas~\ref{ugly_lemma},
\ref{pairwise_comparison}, and \ref{S-recog}.
\end{proof}

\section*{Acknowledgments}
The work contained in the paper came about in response to a question
posed by Jacques Sakarovitch during Michel Rigo's presentation of his
thesis required for the ``habilitation \`a diriger des recherches'' in
France.  We thank Jacques Sakarovitch for his question and we thank
Michel Rigo for presenting the problem to us.


\begin{thebibliography}{99}
\bibitem{AS10}
P.-Y. Angrand and J. Sakarovitch, Radix enumeration of rational
languages.  \textit{RAIRO: Theoret. Informatics Appl.} {\bf 44} (2010)
19--36.

\bibitem{BFRS07} V. Berth\'e, C. Frougny, M. Rigo and J. Sakarovitch, On
  the cost and complexity of the successor function. In {\it
    Proceedings of WORDS 2007} (CIRM, Luminy, Marseille), P. Arnoux,
  N. B\'edaride, J. Cassaigne (Eds.).

\bibitem{BHMV94} V. Bruy\`ere, G. Hansel, C. Michaux and
  R. Villemaire, Logic and p-recognizable sets of integers.
  \textit{Bull. Belg. Math. Soc.} {\bf 1} (1994) 191--238.

\bibitem{CKR10}
\'E. Charlier, T. K\"arki and M. Rigo, Multidimensional generalized automatic
sequences and shape-symmetric morphic words. \textit{Discrete Math.}
{\bf 310} (2010) 1238--1252.

\bibitem{Cob69} A. Cobham, On the base-dependence of set of numbers
  recognizable by finite automata. \textit{Math. Systems Theory} {\bf
 3} (1969) 186--192.

\bibitem{Eilenberg} S. Eilenberg, {\it Automata, languages, and
      machines}, Vol. A. Pure and Applied Mathematics, Vol. 58,
    Academic Press , New York (1974).

\bibitem{EES69} S. Eilenberg, C.C. Elgot and J.C. Shepherdson, Sets
    recognised by $n$-tape automata. \textit{J. Algebra} {\bf 13} (1969)
    447--464.

\bibitem{FS93} Ch. Frougny and J. Sakarovitch, Synchronized relations of
  finite and infinite words. \textit{Theoret. Comput. Sci.} {\bf 18}
  (1993) 45--82.

\bibitem{CANT10} Ch. Frougny and J. Sakarovitch, Number representation
  and finite automata.  In V. Berth\'e, M. Rigo (eds.), {\it Combinatorics,
    Automata, and Number Theory}, Encyclopedia of Mathematics and its
  Applications 135, Cambridge (2010).

\bibitem{GS66}
S. Ginsburg and E.H. Spanier, Semigroups, Presburger formulas and
languages. \textit{Pacific J. Math.} {\bf 16} (1966) 285--296.

\bibitem{LR01} P. Lecomte and M. Rigo, Numeration systems on a regular
  language. \textit{Theory Comput. Syst.} {\bf 34} (2001) 27--44.

\bibitem{CANT10b} P. Lecomte and M. Rigo, Abstract numeration systems.
  In V. Berth\'e, M. Rigo (eds.), {\it Combinatorics,
    Automata, and Number Theory}, Encyclopedia of Mathematics and its
  Applications 135, Cambridge (2010).

\bibitem{RA05} J. Ram\'irez Alfons\'in, {\it The Diophantine Frobenius
    problem}, Oxford Lecture Series in Mathematics and its
  Applications 30, Oxford (2005).

\bibitem{RM02}
M. Rigo and A. Maes, More on generalized automatic
sequences. \textit{J. Autom. Lang. and Comb.} {\bf 7} (2002) 351--376.

\bibitem{Rub04} S. Rubin, {\it Automatic Structures}, Ph.D.\ thesis.
  University of Auckland, New Zealand (2004).

\bibitem{Sem77} A. L. Semenov, The Presburger nature of predicates
  that are regular in two number
  systems. \textit{Sibirsk. Math. \v{Z}.} {\bf 18}
  (1977) 403--418, 479 (in Russian).  English translation in \textit{Siberian
  J. Math.} {\bf 18} (1977) 289--300.
\end{thebibliography}
\end{document}